\documentclass[11pt]{article}
\usepackage{amsfonts}
\usepackage{alltt}
\usepackage{amssymb,verbatim,ifthen}
\usepackage{dsfont}
\usepackage[all]{xy}
\usepackage{color}
\usepackage{graphicx}
\usepackage{amsmath}
\usepackage{amsthm}
\usepackage{mathabx}
\usepackage{pdfsync}
\usepackage{caption}
\usepackage{natbib}

\def\split{true}
\ifthenelse{\equal{\split}{true}} {
\textwidth = 440pt \oddsidemargin = 2pt \evensidemargin = 2pt
}
{}
\newcommand{\ignore}[1]{}

\renewenvironment{proof}[1][Proof.]{\textbf{#1} }{\ \rule{0.5em}{0.5em}}

\newtheorem{theorem}{Theorem}
\newtheorem*{theorem*}{Theorem}

\newtheorem{lemma}{Lemma}

\newtheorem{definition}{Definition}

\newtheorem{remark}{Remark}

\def\1{{1\hskip-2.5pt{\rm l}}}

\def\s{\sigma}

\def\D{\Delta}

\def\G{\Gamma}

\def\Pr{\textbf{{\rm Pr}}}
\def\R{\mathbb{R}}

\def\L{\lambda}

\def\G1{G_{U\mbox-L}}

\def\si{\s_i^f}
\def\spa{\s_{1}^f}

\def\sf{\s^f}
\def\Gf{G_f}
\def\f{f^*}

\makeatother

\begin{document}

\title{On the Failures of Bonus Plans\thanks{The authors wish to thank  Rani Spiegler and Eddie Dekel for their valuable comments. Lehrer acknowledges the support of the Israel Science Foundation, Grant
\#762/045. Lagziel acknowledges the support of the Israel Science Foundation, Grant \#538/11.}}
\author{David Lagziel\thanks{School of Mathematical Sciences,
Tel Aviv University, Tel Aviv 69978, Israel. e-mail:
\textsf{davidlag@post.tau.ac.il}.} \ and  Ehud Lehrer\thanks{School
of Mathematical Sciences, Tel Aviv University, Tel Aviv 69978,
Israel and INSEAD, Bd. de Constance, 77305 Fontainebleau Cedex,
France. e-mail: \textsf{lehrer@post.tau.ac.il}.}  }
\maketitle

\thispagestyle{empty}

\lineskip=2pt\baselineskip=5pt\lineskiplimit=0pt

\noindent{\textsc{Abstract}}:  \begin{quote}
A decision maker (DM) has some funds invested through two investment firms. She wishes to allocate additional funds according to the firms' earnings. The DM, on the one hand, tries to maximize the total expected earnings, while the firms, on the other hand, try to maximize the overall expected funds they manage. In this paper we prove that, for every market, the DM has an optimal bonus policy such that the firms are motivated to act according to the interests of the DM. On the other hand, we also prove that the only policy that is optimal in every market, is independent of the actions and earnings of the firms.
\end{quote}

\bigskip
\noindent {\emph{Journal of Economic Literature} classification
numbers: C72, D47, G11, G14, G23, G24.

\bigskip

\noindent Keywords: Bonus plans; Investments decisions; Market design; Portfolio management.

\newpage
\lineskip=1.8pt\baselineskip=18pt\lineskiplimit=0pt \count0=1

\section{Introduction}
Bonus plans have an important rule in motivating agents in competitive markets, and especially in financial markets. In theory, these plans are intended to guarantee that financial managers are investing their best time and effort to produce the optimal possible earnings for their investors.

Though the motivation for these bonuses is clear, their effects are relatively vague. When the outcomes are stochastic and the plans are based on past performance, the agents\footnote{We sometimes refer to the firms as agents, and to the decision maker as investor.} can take unnecessary risks, from the investors' perspectives, to increase their own expected payoff. Therefore, the formulation of incentives plans has to ensure that agents are acting according to the best interests of the people they represent.

In the current paper we address this problem through a basic set-up. We assume that a decision maker (DM) wishes to invest additional funds using several firms. To do so, she chooses an \emph{allocation rule}, a \emph{bonus plan}, that is a function of the firms' outcomes. For example, a DM has \$100 thousand invested through two firms and she wishes to allocate an additional amount of \$50 thousand as a function of the earnings they produce by the end of the year. Another example - the DM is a manager in an investment firm that needs to distribute additional funds to several departments based on their annual return.

The motivation of the firms to maximize the total expected funds they manage could pose a problem, as it may not coincide with the DM's wishes to maximize the expected return. As it may occur in real-life, the firms can take advantage of the imperfect monitoring of their actions to increase their own expected payoffs, at the expense of the DM.

The results of this paper capture this problem in two aspects. First, we prove that for every market, i.e., for every set of possible actions of the firms, the DM can fix an \emph{optimal bonus plan} such that the firms, in equilibrium, act according to the preferences of the DM. This result applies when the market is relatively \emph{stationary}, in the sense that the set of actions does not change drastically. The proof we present in constructive and holds for a general number of firms and actions. A bonus plan that is linear in the sum of differences of the earnings of the firms proves to be optimal.

On the other hand, we also prove that a universal, non-trivial, optimal bonus plan does not exist. In other words, when the market changes frequently and the set of actions can change considerably, one cannot devise a bonus plan which remains optimal, unless it is independent of the outcomes of the firms.

As these results indicate, we focus on the existence of an optimal bonus plan, given that a bonus plan is needed, rather than discussing the necessity of the plan itself. We assume that agents who are indifferent between the outcomes they produce, due to the lack of incentives, will almost surely lead to a sub-optimal result. For example, firms without results-based payoffs, either positive or negative, will have no reason to invest in R\&D, expand to new markets, gain private information or even improve their trading capabilities.

This perspective defers our research and previous papers. In most previous works (see Subsection \ref{Subsection - Related literature} for a detailed description), the main question is whether a DM can design an incentives plan such that skilled agents are rewarded while unskilled agents are screened out. We, however, show that even when all the agents are skilful, every results-based bonus plan can lead to a sub-optimal outcome. Thus, individual investing capabilities or distinct private information, will have no positive effect on the agents' chosen portfolios.

Another aspect that differs our work and previous ones is the usage of a benchmark portfolio. As we compare the actual outcome of every agent to the potential outcome, given the agent's abilities and information, we do not require an exogenous benchmark. Nonetheless, we show that a bonus plan which is linear w.r.t. the sum of the differences in revenue of the agents is optimal, meaning that a comparison between agents is essential.

The implications of our results are not limited to the motivating example of additional funds divided between competing firms. Our model also applies to cases where an investor wishes to reallocate her previously-invested funds, according to future performance. Therefore, the allocation rule the investor uses, could imply actual bonuses, as well as penalties, and any combination of the two.

\subsection{Related literature} \label{Subsection - Related literature}
Numerous studies were conducted over the importance of performance-based payoffs and reputation in non-deterministic markets. Yet, to the best of our knowledge, the aspect we examine was never discussed. 

Many papers examined the well-known "herding" phenomenon, where agents tend to mimic other agents, although they can preform better with a certain probability, in order to conserve their reputation. One can track back this concept to the words of \cite{keynes2006general}: "Worldly wisdom teaches that it is better for reputation to fail conventionally than to succeed unconventionally." (Book 4, Chapter 12, page 158).

\cite{Scharfstein_Stein_1990} follows this argument and shows, under certain circumstances, that agents simply mimic the behaviour of others, ignoring substantive information they posses. Specifically, when the capabilities of an agent are assessed by the market according to his performance, as well as his conventionality, the agents may mimic each other, leading to a sub-optimal outcome.

\cite{dasgupta2006financial} continues this line of work. It studies the effect of career concerns on the performance of fund managers. Given a specific model, Dasgupta and Prat show that without career concerns, only fund managers with special abilities trade. However, they also prove that career concerns may lead to an equilibrium where unskilled fund managers have to trade in order to stand out, which may involve taking unnecessary risks. This work is generalized to a dynamic model in \cite{dasgupta2008information}. However, the latter focuses on the effect of career concerns on market prices and trading volume. In both cases the investors are risk-neutral return maximizers.

Although the "herding" phenomenon is not our main focus in this paper, it is consistent with our results. We prove that a potentially suboptimal outcome exists, because of the need of agents to level their performance with the ones of their competitors. 

In general, the manipulability abilities of unskilled agents were proven to exist in more than a few papers, such as \cite{Lehrer2001,Sandroni2003,Sandroni2003a,Shmaya2008}; and \cite{Olszewski2008}.

\cite{foster2010gaming} proves that it is almost impossible to generate a bonus plan such that skilled agents are rewarded while unskilled ones are eliminated from the market. This result is based on the assumption that the agents' strategies and tactics are not observable. Recently, \cite{He2015} showed that this result could be altered when a liquidation boundary is set along with requiring the agents to deposit their own money to offset the potential losses.

\cite{Holmstrom1991} studies the principle-agent problem in a completely different set-up than ours, yet it reaches a similar conclusion that sometimes the only optimal incentives plan is constant. It proves that  agents would divert their effort to where it is easier to measure their performance in order to increase their payoff and derives implications with respect to job design. For example, teachers that receive bonuses according to the students' test scores, might neglect other important aspects. 

Assuming that some actions are known to the risk-neutral investor, \cite{Carroll2015} examines contracts where the investor is evaluating the performance of the agent through a worst-case criterion consistent with her own knowledge. Under these conditions, it shows that linear contracts are optimal. Nevertheless, one should not confuse this linearity and our proposed linear bonus plan. In \cite{Carroll2015} the conclusion is that a fixed share of the return is optimal, whereas we suggest a contract where the linearity is taken with comparison to other agents.

\subsection{Outline of the paper}
The paper is organized as follows. Section \ref{Section - example} presents a simple $2$-firm bonus-plan problem that illustrates the drawbacks of results-based incentives in competitive non-deterministic markets. In section \ref{Section - The model} we present the model along with the main assumptions. Section \ref{Section - Main results} includes the main results divided into two parts: In Subsection \ref{Subsection - Fixed set of actions} we show how to formulate an optimal bonus plan for any specific market; In Subsection \ref{Subsection - A universal bonus plan} we prove that every bonus plan can fail as the market evolves, unless it is independent of the outcomes. Concluding remarks and additional comments are given in Section \ref{Section - Conclusions}. 

\section{A motivating example - a $2$-firm bonus-plan problem} \label{Section - example}
A decision maker (DM) wishes to invest some funds through an investment firm. There are two possible investments firms, Firm $1$ and Firm $2$. The DM has already some funds invested through both of the firms and she wishes to allocate the additional funds according to some predetermined rule that depends on the firms' yearly earnings.

The goal of the DM is to maximize the expected earnings of past and future investments. However, as she is not aware of the possible bonds in the market, she chooses to allocate the entire available amount to the company which presents the highest earnings by the end of the following year. In case both firms present the same earnings, the funds are divided equally between the two firms.

The firms can invest either in Bond $X_1$, which gives $5\%$ per year with probability (w.p.) $1$, or in Bond $X_2$, which gives $5.1\%$ per year w.p. $0.6$ and $0\%$ per year w.p. $0.4\%$. The goal of the firms is to maximize their expected earnings and to maximize the overall funds they manage, according to the utility functions given below. 

In this example we prove that, although investing in Bond $X_2$ is suboptimal in expectation relative to an investment in Bond $X_1$ (and also riskier), the unique equilibrium is when both firms invest in Bond $X_2$.

Formally, let $A=\{X_1,X_2\}$ be the set of actions of the firms. We consider the actions $X_1$ and $X_2$ to be random variables on the same probability space with the distributions,
\begin{equation*} \label{Eq - Example bonds}
X_1 =  1.05 \ \mbox{per year w.p. } 1,
\ \ \ \
X_2=
\begin{cases}
1.051,  & \mbox{per year w.p. } \frac{3}{5}, \\
1.0,  & \mbox{per year w.p. } \frac{2}{5}.
\end{cases}
\end{equation*}
These distributions are common knowledge of the firms. Fix $\L \in (0,1)$ and define $U_1$ and $U_2$ to be the utility functions of Firm $1$ and Firm $2$ respectively such that
$$
U_1(X_i,X_j)= \L X_i +(1-\L )\left[ \mathbf{1}_{\{X_i>X_j\}} + \dfrac{\mathbf{1}_{\{X_i=X_j\}}}{2}\right],
$$
and
$$
U_2(X_i,X_j)= \L X_j +(1-\L )\left[ \mathbf{1}_{\{X_j>X_i\}} + \dfrac{\mathbf{1}_{\{X_j=X_i\}}}{2}\right],
$$
for every action $X_i$ of Firm $1$ and for every action $X_j$ of Firm $2$. That is, the utility functions are a weighted average of the earnings and the additional funds distributed by the DM. The following lemma shows that the optimal action of both firms is $X_2$.

\begin{lemma} \label{Lemma - example}
The unique equilibrium is when both firms choose to invest only in Bond $X_2$.
\end{lemma}

Lemma \ref{Lemma - example} exemplifies an interesting situation. The expected earnings per year of Bond $X_2$ is $3.006\%$  while the expected earnings per year of Bond $X_1$ is $5\%$. Therefore, all the parties involved, firms and the DM, would benefit in case both firms choose to invest in $X_1$. Nevertheless, the unique equilibrium is $(X_2,X_2)$ and this market fails to reach an optimal outcome.

\begin{proof}
Let us compute the expected utility of Firm $1$ for every profile of pure actions. There are four profiles we need to consider and the expected utility of Firm $1$ in each is:
\begin{eqnarray*}
E[U_{1}(X_1,X_1)] & = & 1.05 \L + 0.5(1-\L) = 0.5 + 0.55\L; \\
E[U_{1}(X_2,X_1)] & = & 1.05\L  + 0.6(1-\L) = 0.6 + 0.45\L;  \\
E[U_{1}(X_1,X_2)] & = & 1.0306\L  + 0.4(1-\L) = 0.4 + 0.6306 \L;  \\
E[U_{1}(X_2,X_2)] & = & 1.0306\L  + 0.5(1-\L) = 0.5 + 0.5306 \L.
\end{eqnarray*}
A similar computation holds for Firm $2$. Inserting these expected values to a $2$-player game the game presented in Table \ref{Table - Example}.

\begin{table}[h!]
\begin{center}
\begin{tabular}{c|c|c|} 
   & $X_1$ & $X_2$ \\
  \hline
  $X_1$ & $0.5 + 0.55\L$, $0.5 + 0.55\L$ & $0.4 + 0.6306 \L$ , $0.6 + 0.45\L$ \\
  \hline
  $X_2$ & $0.6 + 0.45\L$, $0.4 + 0.6306 \L$ & $0.5 + 0.5306 \L$, $0.5 + 0.5306 \L$ \\
  \hline
\end{tabular}
\end{center}
\caption{The game played by Firm $1$ and Firm $2$ as a function of $\L$.}\label{Table - Example}
\end{table}

For every $\L \in (0,1)$ and for every Firm $i\in \{1,2\}$, action $X_2$ strongly dominates action $X_1$ and the result follows.
\hfill
\end{proof}

This example illustrates the problem when there are two firms and two pure actions. The model, presented in Section \ref{Section - The model}, relates to a more general case in terms of firms and actions.

\section{The model} \label{Section - The model}
Fix $n \in \mathbb{N}$ and define $N=\{1,\dots,n\}$ to be a set of indices. For every $i\in N$, let $X_i$ be a random variable with a countable support $S_i$ and finite expectation defined on the probability space $\Omega$. We refer to $X_i$ as a \emph{pure action}. Denote the set of pure actions $X_1,\dots,X_n$ by $A=\{X_1,\dots,X_n\}$ and denote the union of the supports of the random variables in $A$ by $S_A=\bigcup_{i\in N}S_i$.

A \emph{mixed action} (or \emph{strategy}) is a convex combination of random variables in $A$ such that $q = \sum_{i=1}^n q_i X_i$ and $(q_1,\dots,q_n)\in \D(A)$.\footnote{We use $\D(B)$ to denote the set of probability measures on the set $B$.}  Note that a mixed action $q$ is a random variable and $q(\omega) = \sum_{i=1}^n q_i X_i(\omega)$ for every $\omega \in \Omega$.

Fix a natural $k \geq 2$. A \emph{bonus plan} (BP) of dimension $k$ is a function $f:\R^k \to \R^k$ such that for every $r \in \mathbb{R}^k$,
\begin{eqnarray}\label{eq:fixed sum} \sum_{i=1}^k f_i(r)=1
\end{eqnarray}
 and $f(r) \in [0,1]^k.$
A $k$-player bonus-plan problem consists of a DM and $k$ players (or firms). The DM, who is not familiar with the pure actions in $A$, publicly commits to a bonus plan\footnote{We will not henceforth relate to the dimension of the function $f$ when it is clear from the context.} $f$, which defines a $k$-player game $\Gf$ as follows. Every Player $i$ chooses a strategy, denoted by $\si$, based on the set of pure actions $A$. Denote $\s^f= (\spa,\dots,\s^f_k)$. An $\omega \in \Omega$ is randomly chosen and each Player $i$ receives $f_i\left(\s^f(\omega)\right)$, while the DM receives $\sum_{i=1}^k\si(\omega)$.

\begin{definition}
A profile of strategies $\sf$ is a \emph{Nash equilibrium} in $G_f$ if
\begin{equation}
E\left[f_i\left(\si,\s^f_{-i}\right)\right] \geq E\left[f_i\left(q,\s^f_{-i}\right)\right], \nonumber
\end{equation}
for every Player $i$ and for every strategy $q$.
\end{definition}

A convenient way to consider a $k$-player bonus-plan problem is by assuming that the players are playing a game induced by the DM, which implies that the DM is essentially a game designer or a mechanism designer.
The players in $\Gf$ wish to maximize their expected payoffs in equilibrium, that depend on the bonus plan $f$, and the DM wants to maximize her expected payoff which is a function of the equilibrium strategies of the players in $\Gf$. Formally, Player $i$ wishes to maximize $E\left[f_i\left(\sf\right) \right]$  where $\sf$ is a Nash equilibrium in $G_f$, while the DM wishes the maximize  $E\left[\sum_{i=1}^k\si\right]$.

The bouns-plan problem has some similarities to the principle-agent problem. The DM represents the principle who is interested in motivating her agents (e.g., investments firms) to produce optimal expected earnings. However, the example given in Section \ref{Section - example} shows that the appropriate method is not always the simple one. Sometimes a simple or an intuitive bonus plan could generate a sub-optimal result for all sides. A general result motivated by this example is given in the following section.

\begin{definition}
A bonus plan $f$ is \emph{optimal}, if the induced $k$-player game $\Gf$ has an equilibrium $\sf$ such that
\begin{equation} \label{Eq - Optimality condition}
E\left[\sum_{i=1}^k\si\right] = k \max_{i\in N} E[X_i]. 
\end{equation}
\end{definition}

That is, $f$ is optimal if there exist an equilibrium $\sf$ in $\Gf$  where the mixed actions of all the players are weighted averages of pure actions with maximal expected value.

\begin{remark}
In order to simplify the computations and notations, we assume that
\begin{equation}\label{Ineq - RVs Expectation condition}
E[X_1]> E[X_i],
\end{equation}
for every $2\leq i\leq n$. This implies that a bonus plan $f$ is optimal if and only if $(X_1,\dots,X_1)$ is an equilibrium in $\Gf$.
\end{remark}

\begin{remark}
For every bonus plan $f$ and for every vector of pure actions $X$, it follows that
$$E\left[\sum_{i=1}^kf(X)\right]=1.$$
Therefore, the game $\Gf$ is a fixed-sum game in the sense that the sum of the expected payoffs of all the players is $1$. For every bonus plan $f$ and for every game $\Gf$ we can define an auxiliary bonus plan $\f$ such that $\f(r) = f(r) -\left(\frac{1}{k},\dots,\frac{1}{k}\right)$ for every $r \in \R^k$. The induced game $G_{\f}$ is a symmetric zero-sum $k$-player game.
\end{remark}
\section{Main results} \label{Section - Main results}
In this section we prove the two main results of the paper. The first result shows that for every finite set of random variables $A$, there exists an optimal bonus plan $f$. The second result states that for every non-trivial (i.e., non-constant) bonus plan $f$, there exists a set $A$ such that $f$ is not optimal.

The combination of these results is significant when considering contracts planning. On the one hand, in any non-deterministic market the DM can design a contract that ensures her interests are kept by the agents. On the other hand, if the market is dynamic, meaning that the market is constantly changing in terms of possible actions, then any bonus plan which is not trivial can eventually lead to a suboptimal result. In other words, the only bonus plan that assures that the players are acting according to the DM's preferences, is a bonus plan which is, paradoxically, independent of the players' actions.

\subsection{Fixed set of actions} \label{Subsection - Fixed set of actions}
The first theorem we prove is constructive. For every set of actions $A$, we specify an optimal bonus plan $f$. The optimal bonus plan $f$, that we propose, is linear in the differences of the earning of the players.

We start with the simple case of a bounded set $S_A$. Fix a set of pure actions $A$. Let $I_A$ be the minimal, finite, and closed interval containing all the values of $S_A$ and let $M_A$ be the absolute value of the maximal element in $I_A$ (w.r.t the absolute value). That is, for every $\omega \in \Omega$ and every action $X_i$, it follows that $X_i(\omega) \in I_A$ and $|X_i(\omega)|\leq M_A$. Define the $M_A$-\emph{Linear Bonus Plan} $f$ by
\begin{equation}
f_{i}(r)=\nonumber
\begin{cases}
\frac{1}{k} +\frac{\sum_{j=1}^k (r_i-r_j)}{2k(k-1)M_A}, &\mbox{if } r \in I_A^k, \\
\frac{1}{k}, & \mbox{if } r \notin I_A^k.
\end{cases}
\nonumber
\end{equation}
One can verify that $f$ is well-defined, since for every $r \in \R^k$, the equality $\sum_i f_i(r) =1$ holds and $f(r) \in [0,1]^k$.
\begin{theorem} \label{Theorem - LAR is optimal}
For every set of pure actions $A$ with finite support $S_A$, the $M_A$-Linear Bonus Plan is optimal.
\end{theorem}
\begin{proof}
We have to prove that $(X_1,\dots,X_1)$ is an equilibrium in $G_f$, or equivalently, that for every Player $i$ and for every strategy $q = \sum_{i=1}^n q_i X_i$ of Player $i$ (when $q\neq X_1$), the inequality $E\left[f_i\left(X_1,\dots,X_1\right)\right] < E\left[f_i\left(q,X_1,\dots,X_1\right)\right] \nonumber $ holds. Without loss of generality, assume that $i=1$, therefore
\begin{eqnarray}
E\left[f_1\left(q,X_1,\dots,X_1\right)\right] & = & E\left[f_1\left(\sum_{i=1}^n q_i X_i,X_1,\dots,X_1\right)\right]  \nonumber \\
& = & E\left[\frac{(k-1)\sum_{i=1}^n q_i X_i-(k-1)X_1}{2k(k-1) M_A} + \frac{1}{k} \right]  \nonumber  \\
& = & \frac{1}{2k M_A} E\left[\sum_{i=1}^n q_iX_i - X_1 \right] + \frac{1}{k}  \nonumber \\
& < & \frac{1}{k} = E\left[f_i\left(X_1,\dots ,X_1\right)\right], \nonumber
\end{eqnarray}
when the second and last equalities follow from the definition of $f$ and the inequality follows from the fact that $q\neq X_1$ and $E[q] = E\left[\sum_{i=1}^n q_i X_i \right]  < E[X_1]$. \hfill
\end{proof}

The fact that $(X_1,\dots,X_1)$ is an equilibrium is not surprising. It follows directly from the linearity of $f_i(r)$ in $r_i$. This linearity implies that $X_1$ is a dominant strategy for every player and therefore, the equilibrium is unique.

The following theorem presents a Bounded and Linear Bonus Plan $f$ for the general case (i.e., for the case that $S_A$ is not bounded) and proves that it is optimal.

\begin{theorem} \label{Theorem - BLAR is optimal}
For every set of pure actions $A$, there exists an optimal bonus plan.
\end{theorem}
\begin{proof}
We extend Theorem \ref{Theorem - LAR is optimal} for an unbounded set $S_A$ by constructing a new optimal bonus plan. Let $\hat{r}_{q,i}=(q,(X_1)_{j\neq i})$ be a vector of actions where player $i$ chooses the mixed action $q$ and every other player plays $X_1$. Note that for every $q\neq X_1$, the expected value $E\left[q-X_1\right]=-c_q <0$ for some $c_q>0$. Since the expected value is negative and finite, there exists a real number $M(q)>0$ such that for every $m \geq M(q)$
\begin{equation} \label{Ineq - The expectation with M(q) 1}
E\left[(q-X_1) \mathbf{1}_{\{|q-X_1|\leq m\}}\right] < -\frac{c_q}{2},
\end{equation}
and
\begin{equation} \label{Ineq - The expectation with M(q) 2}
\sum_{|l|>m}|l|\Pr(q-X_1=l)<\frac{c_q}{2}.
\end{equation}

Ineq. \eqref{Ineq - The expectation with M(q) 1} and \eqref{Ineq - The expectation with M(q) 2} are strict, therefore there exists an open set $B(q)$ containing $q$ such that both inequalities hold for every $q' \in B(q)$ where $q'\neq X_1$. The set of mixed actions $q$ is compact as it can be represented by a $(n-1)$-simplex, hence the set of open sets $\{B(q)\}$ is an open cover. It is known that there exists a finite subcover $B_A$. Fix $c = \min_{q\in B_A}c_q > 0$ and let $M \geq \max_{q \in B_A} M(q)$ be a positive number such that for every $m \geq M$ and every $q\neq X_1$,
\begin{equation} \label{Ineq - The expectation with M(q) 3}
E\left[(q-X_1) \mathbf{1}_{\{|q-X_1|\leq m\}}\right] < -\frac{c}{2},
\end{equation}
and
\begin{equation} \label{Ineq - The expectation with M(q) 4}
\sum_{|l|>m}|l|\Pr(q-X_1=l)<\frac{c}{2}.
\end{equation}
Note that  \eqref{Ineq - The expectation with M(q) 4} hold for $q=X_1$ with weak inequality.

Define the \emph{Bounded-Linear Bonus Plan} $f$ by
\begin{equation}
f_{i}(r)=\nonumber
\begin{cases}
\frac{1}{k} +\frac{\sum_{j=1}^k (r_i-r_j)}{2k(k-1)M}, &\mbox{if \ } 0\leq \frac{1}{k} +\frac{\sum_{j=1}^k (r_i-r_j)}{2k(k-1)M}\leq \frac{2}{k}, \\
\frac{1}{k}, & \mbox{if \ } \rm{otherwise}.
\end{cases}
\nonumber
\end{equation}
One can verify that $f$ is a well-defined bonus plan.

We prove that $(X_1,\dots,X_1)$ is an equilibrium in the game induced by the Bounded-Linear Bonus Plan by showing w.l.o.g. that $$
E[f_1(q,X_1,\dots,X_1)] < \frac{1}{k} = E[f_1(X_1,\dots,X_1)],
$$
for every $q \neq X_1$. Note that
\begin{eqnarray}
0 \leq \frac{1}{k} +\frac{\sum_{j\neq i} (q-X_1)}{2k(k-1)M}\leq \frac{2}{k} \ \  & \Leftrightarrow & \ \  -2M \leq q-X_1 \leq 2M \nonumber \\ \nonumber
& \Leftrightarrow & \ \ |q-X_1|\leq 2M.
\end{eqnarray}
Thus,
\begin{eqnarray}
E[f_1(q,X_1,\dots,X_1)] & = & E\left[\left(\frac{1}{k} + \frac{\sum_{j\neq i} (q-X_1)}{2k(k-1)M}\right) \mathbf{1}_{\{|q-X_1| \leq 2M \}} \right]  \nonumber \\
& + & \frac{1}{k} \Pr(|q-X_1| > 2M) \nonumber \\
& \leq & \frac{1}{k} + \frac{1}{2M k} E\left[(q-X_1) \mathbf{1}_{\{|q-X_1| \leq 2M\}} \right] \nonumber \\
& + & \frac{1}{k}\sum_{|l|>2M}\frac{|l|}{2M} \Pr(q-X_1=l). \nonumber
\end{eqnarray}
It follows from \eqref{Ineq - The expectation with M(q) 3} and \eqref{Ineq - The expectation with M(q) 4} that
\begin{eqnarray}
E[f_1(q,X_1,\dots,X_1)] & < & \frac{1}{k} -\frac{c}{4Mk} + \frac{1}{2Mk}\sum_{|l|>2M}|l|\Pr(q-X_1=l) \nonumber \\
& < &  \frac{1}{k} -\frac{c}{4Mk} + \frac{c}{4Mk} = \frac{1}{k}, \nonumber
\end{eqnarray}
as previously stated.
\hfill
\end{proof}

\subsection{A universal bonus plan} \label{Subsection - A universal bonus plan}

A \emph{universal bonus plan} $f$ is a non-constant bonus plan such that for every finite set of actions, $f$ is optimal. Specifically, A universal bonus plan $f$ is a non-constant bonus plan such that for every set of actions $X_1,\dots,X_n$, where $E[X_1] >E[X_j]$ for every $2\leq j\leq n$, the profile of actions $(X_1,\dots,X_1)$ is an equilibrium in $G_f$. The following theorem proves such a plan does not exist.

\begin{theorem} \label{Theorem - Imopssibility theorem}
If there are two players, a universal bonus plan does not exist.
\end{theorem}

\begin{proof}
Let $x<y<z$. We first prove that  $f_1(y,y)\ge f_1(x,y)$.
Assume to the contrary that $f_1(y,y)< f_1(x,y)$. Let $A= \{X_1,X_2\}$ be a set of two actions $X_1$ and $X_2$ with a joint probability distribution,
\begin{center}
  \begin{tabular}{ | c || c | c |}
    \hline
    $X_1 \backslash X_2$ & x & y \\ \hline \hline
    x & 0 & 0 \\ \hline
    y & 1 & 0 \\
    \hline
  \end{tabular}
\end{center}
Note that $E[X_1]>E[X_2]$. However,
$$E[f_1(X_2,X_1)] = f_1(x,y) > f_1(y,y) = E[f_1(X_1,X_1)],$$
implying that  $(X_1,X_1)$ is not an equilibrium in $G_f$, since Player $1$ can benefit from deviating to $X_2$. Thus,
\begin{equation} \label{Ineq 1 - General theorem for k=2}
f_1(y,y)\ge f_1(x,y).
\end{equation}
For similar reasons,
\begin{equation} \label{Ineq 2 - General theorem for k=2}
f_2(y,y)\ge f_2(y,x).
\end{equation}

Next we prove that $f_1(x,x)\ge f_1(y, x).$ Let $p$ be a number in $(0,1)$ and let $A= \{X_1,X_2\}$ be a set of two actions $X_1$ and $X_2$ with a joint probability distribution,
\begin{center}
  \begin{tabular}{ | c || c | c |}
    \hline
    $X_1 \backslash X_2$ & x & y \\ \hline \hline
    x & 0 & $p$  \\ \hline
      z & $1-p$ & 0  \\ \hline
  \end{tabular}
\end{center}
A direct computation shows that
\begin{eqnarray}
E[f_1(X_1,X_1)]-E[f_1(X_2,X_1)] & = & p (f_1(x,x) - f_1(y,x)) \nonumber \\
&  + & \left(1-p \right)(f_1(z,z)-f_1(x,z)). \nonumber
\end{eqnarray}
Recall that $f_1(z,z)-f_1(x,z)$ is bounded. If $f_1(x,x) < f_1(y,x)$, then there is $p$ smaller than, but sufficiently close to $1$, such that for every $z$, $E[f_1(X_1,X_1)]-E[f_1(X_2,X_1)]<0$. In other words,
\begin{equation} \label{Ineq 5}
E[f_1(X_1,X_1)]<E[f_1(X_2,X_1)].
\end{equation}
Now one can choose $z$ to be sufficiently large, so that $E[X_1]>E[X_2]$. Inequality \eqref{Ineq 5} implies that $(X_1,X_1)$ is not an equilibrium in $G_f$, since Player $1$ can benefit from deviating to $X_2$.  Hence,
\begin{equation} \label{Ineq 3 - General theorem for k=2}
f_1(x,x)\ge f_1(y,x).
\end{equation}
Similar reason imply that
\begin{equation} \label{Ineq 4 - General theorem for k=2}
f_2(x,x)\ge f_2(x,y).
\end{equation}

We now sum up inequalities (\ref{Ineq 1 - General theorem for k=2}), (\ref{Ineq 2 - General theorem for k=2}), (\ref{Ineq 3 - General theorem for k=2}) and (\ref{Ineq 4 - General theorem for k=2}) to obtain,
$f_1(y,y)+f_2(y,y)+ f_1(x,x) + f_2(x,x)\ge f_1(x,y)+ f_2(y,x)+f_1(y,x)+f_2(x,y)$. Due to Eq.\ (\ref{eq:fixed sum}) we obtain equality. Thus,  (\ref{Ineq 1 - General theorem for k=2}), (\ref{Ineq 2 - General theorem for k=2}), (\ref{Ineq 3 - General theorem for k=2}) and (\ref{Ineq 4 - General theorem for k=2}) are actually equalities. Therefore, $$
f_1(x,x) = f_1(x,y) = f_1(y,x) = f_1(y,y),
$$ and the proof is complete. \hfill
\end{proof}

A generalization of Theorem \ref{Theorem - Imopssibility theorem} to any number of players $k\geq 3$ is not trivial. In fact, the requirement that $f$ not be constant is not sufficient. For example, take any non-constant bonus plan $f:\R^k \to \R^k$ such that $f_i(r) = \frac{1}{k}$ for every $r\in \R^k$ with at least two identical coordinates. In this case, for every action $Xj$, the profile of strategies $(X_j,\dots,X_j)$ is an equilibrium, as any deviation of a single player has no influence on the payoffs.  On the other hand, requiring that all equilibria satisfy Eq. \eqref{Eq - Optimality condition} does not hold when $f$ is constant, as all profiles are equilibria. Therefore, we add a new requirement when the number of players is three or more.

\begin{definition}
A strongly universal bonus plan $f$, is a bonus plan where every profile of actions $\sf$ in $G_f$ that sustains Eq. \eqref{Eq - Optimality condition} is an equilibrium.
\end{definition}

\begin{theorem} \label{Theorem - Impossibility for k>2}
If $f$ is a strongly universal bonus plan, then every profile of actions is an equilibrium.
\end{theorem}

\begin{proof}
Let $f$ be a strongly universal bonus plan. Clearly, Theorem \ref{Theorem - Imopssibility theorem} implies that the result holds for the case of $k=2$. Fix $k\geq 3$. We prove the theorem by showing that for every player $i$ and for every vector of outcomes $r\in \R^k$, the $i^{th}$ coordinate $f_i(r)$ of the bonus plan is non-decreasing  and non-increasing in $r_i$.

Assume to the contrary that there exists a player $i$, a vector of outcomes $r \in \R^k$, and $w_i \in \R$ such that $f_i(w_i,r_{-i})> f_i(r_i,r_{-i})$ where $w_i < r_i$. Define the random variable $X$ such that $\Pr(X=x)>0$ iff $x =r_j$ when $1\leq j \leq k$. Assume that $\Pr(X=r_i)>\Pr(X=r_j)$ for every $j\neq i$. In addition, define a set of i.i.d. random variables $X_j \sim X$ where $1\leq j\leq k$. Define the vector-valued random variable $(W,X_{-i})$ by 
$$
\Pr((W,X_{-i}) =x)= \Pr((X_i,X_{-i}) = x) \ \ \forall  x\neq r,
$$
and
$$
\Pr((W,X_{-i}) =(w_i,r_{-i})) = \Pr((X_i,X_{-i}) = r).
$$
Clearly, $(W,X_{-i})$ and $W$ are well-defined. A direct computation shows that $E[W] < E[X]$. However, the vector $(X_i,X_{-i})$ is not an equilibrium as player $i$ can deviate to $W$ and increase his payoff since 
$f_i(w_i,r_{-i})> f_i(r_i,r_{-i})$. Hence $f_i(\cdot,r_{-i})$ is non-decreasing for every $i$ and every $r_{-i}$.

Now assume to the contrary that $f_i(r_i,r_{-i})$ is strictly increasing in $r_i$. That is, there exists a player $i$, a vector of outcomes $r \in \R^k$, and $y_i \in \R$ such that $f_i(y_i,r_{-i})> f_i(r_i,r_{-i})$ where $y_i > r_i$.

Let $\bar{z},\underline{z}\in \R$ be two real number such that $\bar{z}>r_j >\underline{z}$ for every $1\leq j \leq k$ and let $p$ be a number in $(0,1)$. Define the random variable $Y$ such that, with probability $p$, it follows that $\Pr(Y=y)>0$ iff $y =r_j$ when $1\leq j \leq k$ (assume that $\Pr(Y=r_i)>\Pr(Y=r_j)$ for every $j\neq i$), and with probability $1-p$, the random variable $Y$ equals $\bar{z}$. In addition, define a set of i.i.d. random variables $Y_j \sim Y$ where $1\leq j\leq k$. Define the vector-valued random variable $(Z,Y_{-i})$ by 
$$
\Pr((Z,Y_{-i}) =y)= \Pr((Y_i,Y_{-i}) = y) \ \ \forall  y\neq r, y_j\neq \bar{z} \ \forall j,
$$
$$
\Pr((Z,Y_{-i}) =(y_i,r_{-i})) = \Pr((Y_i,Y_{-i}) = r),
$$
and if there exists a coordinate $j$ in $y \in \R^k$ such that $y_j = \bar{z}$, then
$$
\Pr((Z,Y_{-i})=(\underline{z},y_{-i}))= \Pr((Y_i,Y_{-i}) = y).
$$

Clearly, $(Z,Y_{-i})$ and $Z$ are well-defined. Note that
\begin{eqnarray}
E[f_i(Z,Y_{-i})] & = & E\left[f_i(Y_i,Y_{-i})\mathbf{1}_{\{Y_{-i}\neq r, Y_j\neq \bar{z} \ \forall j\}}\right] +f_i(y_i,r_{-i}) \Pr((Y_i,Y_{-i}) = r) \nonumber \\
& + & \sum_{\substack{ y\in \R^k : \\ \exists j, y_j=\bar{z}}}f_i(\underline{z},y_{-i}) \Pr((Y_i,Y_{-i}) = y) \nonumber \\
& > & E\left[f_i(Y_i,Y_{-i})\mathbf{1}_{\{Y_j\neq \bar{z} \ \forall j\}}\right] + \sum_{\substack{ y\in \R^k : \\ \exists j, y_j=\bar{z}}}f_i(\underline{z},y_{-i}) \Pr((Y_i,Y_{-i}) = y) \label{Ineq 1 - General theorem for k>3}  \\
& = & E\left[f_i(Y_i,Y_{-i})\right] + \sum_{\substack{ y\in \R^k : \\ \exists j, y_j=\bar{z}}}(f_i(\underline{z},y_{-i})-f_i(y)) \Pr((Y_i,Y_{-i}) = y), \label{Ineq 2 - General theorem for k>3} 
\end{eqnarray}
when Ineq. \eqref{Ineq 1 - General theorem for k>3} follows from the assumption that $f_i(y_i,r_{-i})> f_i(r_i,r_{-i})$. The sum in Eq. \eqref{Ineq 2 - General theorem for k>3} is bounded, therefore we can choose $p$ sufficiently close to $1$ (but still smaller than $1$), such that for every $\bar{z}$, the inequality $E[f_i(Z,Y_{-i})] > E\left[f_i(Y_i,Y_{-i})\right] $ holds. Taking a sufficiently large $\bar{z}$  and a sufficiently low $\underline{z}$, guarantees that $E[Y]>E[Z]$. 

In conclusion, the vector $(Y_i,Y_{-i})$ is not an equilibrium as player $i$ can deviate to $Z$ and increase his payoff. Contradiction. Hence $f_i(\cdot,r_{-i})$ is non-increasing for every $i$ and every $r_{-i}$. The combination of the two results, proves that $f_i$ is independent of the $i^{th}$ coordinate. This implies that the expected payoff of every player $i$ is independent of his actions and every profile of actions is an equilibrium.
\hfill
\end{proof}

\section{Conclusions and addtional comments} \label{Section - Conclusions}
The results presented in this paper have two completing aspects: an applicative one and a theoretical one. On the one hand, we give a specific description of a bonus plan that guarantees that agents are motivated to act according to the DM's wishes. On the other hand, we show that an optimal bonus plan simply does not exist, unless it is useless as a bonus plan (constant, meaning fixed payoffs).

This work is based on the assumption that the DM tries to maximize her expected payoff. Although, this assumption is common in the literature, one can still follow the line of \cite{Holmstrom1991} and assume that the DM tries to optimize a few elements simultaneously. For example, maximizing the expected revenue while minimizing the risk. In general, one can assume that the DM has a certain preference relation over the set of actions and try to find a bonus plan which implements this preference in equilibria. These problems are left for future research. 

Another element left for future work is the generalization given in Theorem \ref{Theorem - Impossibility for k>2}. This generalization could be done in various ways. E.g., we can assume that $f_i(\cdot,r_{-i})$ always has some dependency on the action of player $i$ for every $r_{-i}$, and then explore whether an optimal bonus plan is possible under such conditions.


\bibliographystyle{aer}
\bibliography{../../MyCollection}


\end{document}